\documentclass[submission,copyright,creativecommons]{eptcs}
\providecommand{\event}{SOS 2007} 

\usepackage{iftex}

\ifpdf
  \usepackage{underscore}         
  \usepackage[T1]{fontenc}        
\else
  \usepackage{breakurl}           
\fi

\usepackage{enumerate,xspace}
\usepackage{amsmath,amssymb,wasysym}
\usepackage[all]{xy}
\usepackage{proof}
\usepackage[svgnames]{xcolor}
\usepackage{tikz}
\usepackage[sort,nocompress]{cite}
\usepackage{stmaryrd} 
\usepackage{mathtools}
\usepackage{latexsym}
\usepackage{dsfont}
\usepackage{multicol}
\usepackage{cmll}
\usepackage{multirow}
\usepackage{longtable}

\newtheorem{observation}{Remark}[section]
\newtheorem{lemma}[observation]{Lemma} 
\newtheorem{theorem}[observation]{Theorem}
\newtheorem{definition}[observation]{Definition}
\newtheorem{example}[observation]{Example}

\newtheorem{corollary}[observation]{Corollary}

\title{Reverse Faà di Bruno's Formula\\ for Cartesian Reverse Differential Categories}
\author{
Aaron Biggin\footnote{For work on this research, this author was funded by a Summer Vacation Research Scholarship from Macquarie University.} \qquad\qquad Jean-Simon Pacaud Lemay\footnote{This material is based upon work supported by the AFOSR under award number FA9550-24-1-0008. This author is also funded by an ARC DECRA (DE230100303).}
\institute{Macquarie University\\
Sydney, New South Wales, Australia}
\email{\quad aaron.biggin@students.mq.edu.au \quad\qquad js.lemay@mq.edu.au}
}
\def\titlerunning{Reverse Faà di Bruno's Formula for CRDC}
\def\authorrunning{A. Biggin \& J.-S. P. Lemay}
\begin{document}
\maketitle
\allowdisplaybreaks

\begin{abstract}
Reverse differentiation is an essential operation for automatic differentiation. Cartesian reverse differential categories axiomatize reverse differentiation in a categorical framework, where one of the primary axioms is the reverse chain rule, which is the formula that expresses the reverse derivative of a composition. Here, we present the reverse differential analogue of Faà di Bruno's Formula, which gives a higher-order reverse chain rule in a Cartesian reverse differential category. To properly do so, we also define partial reverse derivatives and higher-order reverse derivatives in a Cartesian reverse differential category. 
\end{abstract}

\section{Introduction}

There are two types of derivative operations used in automatic differentiation: \textbf{forward differentiation} and the \textbf{reverse differentiation} (also referred to as forward mode and reverse mode differentiation respectively). For a smooth function $F: \mathbb{R}^n \to \mathbb{R}^m$, its forward derivative is the usual total derivative, which is of type $\mathsf{D}[F]: \mathbb{R}^n \times \mathbb{R}^n \to \mathbb{R}^m$ and can be defined using the Jacobian of $F$. On the other hand, its reverse derivative is of type $\mathsf{R}[F]: \mathbb{R}^n \times \mathbb{R}^m \to \mathbb{R}^n$, which can instead be defined in terms of the transpose of the Jacobian of $F$. As such, the forward derivative and the reverse derivative are the transposes of each other. Forward differentiation is more efficient when the dimension of the output is much larger than the dimension of the input, $n \leq m$. In comparison, reverse differentiation is more efficient when the dimension of the input is much larger than the dimension of the output, $n \geq m$. Since many optimization algorithms, especially those used in machine learning, deal with large input data sets that are characterized by functions of type $\mathbb{R}^n \to \mathbb{R}$ where $n$ is quite large, it is reverse mode differentiation which is often more practical in these setting. Motivated by this, there has been significant interest in studying the categorical foundations of reverse differentiation \cite{cockett2019reverse,cruttwell2022monoidal,reverseascent,reversesemantics,cruttwelletal:LIPIcs.CSL.2024.21,Lemaygradient}, which fits in the larger program of providing the categorical foundations for machine learning \cite{vakar2022chad,catML,cruttwell2019towards,gavranovic2024categorical,wilson2022categories,fong2019backprop}. 

Cartesian differential categories \cite{blute2009cartesian} provide the categorical foundations of forward differentiation, as well as the categorical semantics of the differential $\lambda$-calculus. In particular, a Cartesian differential category comes equipped with a \textbf{forward differential combinator} (Sec \ref{sec:forward}), which is an operator that captures forward differentiation of maps. The axioms of a forward differential combinator are analogues of the fundamental identities of the total derivative from differential calculus, such as the famous chain rule for the forward derivative of a composition. On the other hand, Cartesian reverse differential categories provide the categorical foundations of reverse differentiation. This time, a Cartesian reverse differential category comes equipped with a \textbf{reverse differential combinator} (Def \ref{def:CRDC}), which is an operator that captures reverse differentiation of maps. The axioms of the reverse differential combinator are the reverse differential counterparts of the axioms of the forward differential combinator. In particular, one of the axioms is the reverse chain rule for the reverse derivative of a composition. Every reverse differential combinator induces a forward differential combinator as well as a transpose operation, called a \textbf{contextual linear dagger} (Sec \ref{sec:forward}), such that the reverse derivative and forward derivative are transposes of one another \cite[Thm 42]{cockett2019reverse}. As such, every Cartesian reverse differential category is also a Cartesian differential category \cite[Thm 16]{cockett2019reverse}. Therefore, there is interest in studying and developing the reverse differential counterparts of well-known forward differential concepts in a Cartesian reverse differential category. 

Famously, \textbf{Faà di Bruno's Formula} provides a higher-order chain rule for the formula of a higher-order forward derivative of a composition. Faà di Bruno's Formula also holds in a Cartesian differential category \cite[Lemma 3.4]{garner2020cartesian}. The main objective of this paper is to provide the reverse differential version of Faà di Bruno's Formula, giving a higher-order reverse chain rule. As we will see, the reverse Faà di Bruno's Formula surprisingly involves both reverse derivatives and forward derivatives. To properly give the reverse Faà di Bruno's Formula, we will first need to develop the appropriate notions of \textbf{partial reverse derivatives} (Def \ref{def:partialreverse}) and \textbf{higher-order reverse derivatives} (Def \ref{def:higherreverse}) in a Cartesian reverse differential category, in such a way that they are the transposes of their respective forward derivative counterparts. So, we will show that the transpose of the partial forward derivative is indeed the partial reverse derivative (Lemma \ref{lemma:partialtranspose}). However, to also get that the transpose of the higher order forward derivative is the higher order reverse derivative, we introduce an extra compatibility relation between the reverse differential combinator and the forward differential combinator, which we call the \textbf{stable rule} (Def \ref{def:stable}). This additional rule is not very imposing and quite natural, and holds in prominent examples of Cartesian reverse differential categories (Ex \ref{ex:stablesmooth}). Then, in the same way that the reverse chain rule can be computed by taking the transpose of the chain rule, by taking the transpose of Faà di Bruno's Formula, we can compute the reverse Faà di Bruno's Formula (Thm \ref{Prop:faa}). 

To the best of the authors' knowledge, this is the first expression of Faà di Bruno's Formula for reverse mode differentiation. The reverse Faà di Bruno's Formula will hopefully help improve calculations in (categorical) models for automatic differentiation and machine learning. Moreover, an important application of Faà di Bruno's Formula for Cartesian differential categories is that it is a crucial formula in the construction of \emph{cofree} Cartesian differential categories \cite{garner2020cartesian,cockett2011faa}. Therefore, in future work, it would be interesting to understand if the reverse Faà di Bruno's Formula can be used somehow to construct \emph{cofree} Cartesian reverse differential categories. 

 \textbf{Conventions:} In an arbitrary category $\mathbb{X}$, objects will be denoted by capital letters $A,B,C$, etc. and maps by miniscule letters $f,g,h$, etc. We denote hom-sets by $\mathbb{X}(A,B)$, maps as arrows ${f: A \to B}$, identity maps as $1_A: A \to A$, and we use the classical notation for composition $\circ$, as opposed to diagrammatic order that was used in other papers on Cartesian (reverse) differential categories, such as in \cite{blute2009cartesian,cockett2019reverse}. For a category with finite products, we denote the product by $\times$, the projection maps by ${\pi_j: A_1 \times \hdots \times A_n \to A_j}$, and the pairing operation as $\langle -, \hdots, - \rangle$. 

\section{Reverse Differentiation}

In this section, we review Cartesian reverse differential categories and, in particular, also develop the notion of \emph{partial} reverse derivatives. Partial reverse differentiation is very practical and provides a helpful new perspective on the axioms of reverse differentiation and, as we will see in the next section, defining forward differentiation and linear transpose operation from reverse differentiation.

Here, following what was done for forward differentiation in \cite{garner2020cartesian}, we will introduce the definition of a Cartesian reverse differential category relative to a fixed commutative semiring $k$, slightly generalising the original definition from \cite{cockett2019reverse}. As such, in this relative setting, the underlying structure of a Cartesian reverse differential category is that of a \emph{Cartesian left $k$-linear category} \cite[Sec 2.1]{garner2020cartesian}, which can be described as a category with finite products which is \emph{skew}-enriched over the category of $k$-modules and $k$-linear maps between them \cite{garner2020cartesian}. Essentially, this means that each hom-set is a $k$-module, so we have zero maps, can take the sum of maps, and can scalar multiply maps by elements of $k$, but also allow for maps which do not preserve zeroes, sums, or scalar multiplication. Maps which do preserve the module structure are called \emph{$k$-linear maps}. Explicitly, a \textbf{left $k$-linear category} is a category $\mathbb{X}$ such that each hom-set $\mathbb{X}(A,B)$ is a $k$-module with scalar multiplication $\cdot : k \times  \mathbb{X}(A,B) \to  \mathbb{X}(A,B)$, addition ${+: \mathbb{X}(A,B) \times \mathbb{X}(A,B) \to \mathbb{X}(A,B)}$, and zero $0 \in \mathbb{X}(A,B)$; and such that pre-composition preserves the $k$-linear structure: $(s \cdot f + t \cdot g) \circ x = s \cdot (f \circ x) + t \cdot (g \circ x)$. A map $f: A\to B$ is said to be \textbf{$k$-linear} if post-composition by $f$ preserves the $k$-linear structure: $f \circ (s \cdot x + t \cdot y) =   s \cdot (f \circ x) + t \cdot (f \circ y)$. A \textbf{Cartesian left $k$-linear category} is a left $k$-linear category $\mathbb{X}$ such that $\mathbb{X}$ has finite products and all projection maps $\pi_j$ are $k$-linear. We note that when taking $k=\mathbb{N}$, the semiring of natural numbers, (Cartesian) left $\mathbb{N}$-linear categories and their $\mathbb{N}$-linear maps are the same thing as (Cartesian) left additive categories and their additive maps from \cite[Def 1.1.1 \& 1.2.1]{blute2009cartesian}. 

A Cartesian reverse differential category is a Cartesian $k$-linear category which comes equipped with a \emph{reverse differential combinator}, which is an operator that sends maps to their reverse derivative. In this relative to $k$ setting, the reverse differential combinator still satisfies essentially the same seven axioms as in \cite[Def 13]{cockett2019reverse}, but where we upgrade the first two \textbf{[RD.1]} and \textbf{[RD.2]} from simple additivity to $k$-linearity in the obvious way (so \textbf{[RD.3]} to \textbf{[RD.7]} are exactly the same). So, in particular, when taking $k= \mathbb{N}$, a Cartesian $\mathbb{N}$-reverse differential category is precisely the same thing as the original definition of a Cartesian reverse differential category from \cite[Def 13]{cockett2019reverse}. 

\begin{definition}\label{def:CRDC} \cite[Def 13]{cockett2019reverse} A \textbf{Cartesian $k$-reverse differential category} is a Cartesian left $k$-linear category that comes equipped with a \textbf{reverse differential combinator} $\mathsf{R}$, that is, a family of functions $\mathsf{R}: \mathbb{X}(A , B) \to \mathbb{X}(A \times B, A)$, satisfying the seven axioms \textbf{[RD.1]} to \textbf{[RD.7]} described below. For a map $f: A \to B$, the map $\mathsf{R}[f]: A \times B \to A$ is called the \textbf{reverse derivative} of $f$. 
\end{definition}

Before we review the axioms of the reverse differential combinator, it may be helpful to first quickly review the canonical example of a Cartesian reverse differential category: 

\begin{example}\label{ex:smooth} Let $\mathbb{R}$ be the set of real numbers. Define $\mathsf{SMOOTH}$ as the category whose objects are the Euclidean real vector spaces $\mathbb{R}^n$ and whose maps are real smooth functions ${F: \mathbb{R}^n \to \mathbb{R}^m}$ between them. $\mathsf{SMOOTH}$ is a Cartesian $\mathbb{R}$-reverse differential category where the reverse differential combinator is defined using the \emph{transpose} of the Jacobian, see \cite[Ex 14.2]{cockett2019reverse} and \cite[Ex 4.2]{Lemaygradient}. Explicitly, for a smooth function $F: \mathbb{R}^n \to \mathbb{R}^m$ (where recall that $F$ is in fact a tuple $F = \langle f_1, \hdots, f_m \rangle$ of smooth functions $f_j:  \mathbb{R}^n \to \mathbb{R}$), its reverse derivative $\mathsf{R}[F]: \mathbb{R}^n \times \mathbb{R}^m \to \mathbb{R}^n$ is defined as follows: 
\begin{align*}
     \mathsf{R}[F](\vec x, \vec y) := \left(\sum \limits^m_{j=1} \frac{\partial f_j}{\partial x_1}(\vec x) y_j, \hdots, \sum \limits^m_{j=1} \frac{\partial f_j}{\partial x_n}(\vec x) y_j \right) 
\end{align*}
Other examples of Cartesian reverse differential categories can be found in \cite{cockett2019reverse,cruttwell2022monoidal}, which include any dagger category with dagger biproducts \cite[Ex 19]{cockett2019reverse}, polynomials over a fixed semiring \cite[Ex 14.1]{cockett2019reverse}, and the coKleisli category of a monoidal reverse differential category \cite[Thm 55]{cruttwell2022monoidal}.
\end{example}

So what are the seven axioms \textbf{[RD.1]} to \textbf{[RD.7]} of a reverse differential combinator? Simply writing down the equations of \textbf{[RD.1]} to \textbf{[RD.7]} as in \cite[Def 13]{cockett2019reverse} can be a bit non-intuitive to parse at first glance. Especially the last two axioms \textbf{[RD.6]} and \textbf{[RD.7]} are a bit complicated to write down and somewhat intimidating at first. To aid in readability, we will use a term calculus for reverse differentiation inspired by the one for forward differentiation \cite[Sec 4]{blute2009cartesian}. In this term calculus, the reverse derivative $\mathsf{R}[f](a, b)$ will be written as:
\[\frac{\mathsf{r} f(x)}{\mathsf{r}x}(a) \cdot b\] 
We leave formalising this term calculus, including proving soundness and completeness, for future work. So, using this term calculus, the first five axioms of the reverse differential combinator are: 
\begin{enumerate}[{\bf [RD.1]}]
\item The reverse differential combinator is a $k$-linear morphism:
\[\dfrac{\mathsf{r} \left(s \cdot f(x) + t \cdot g(x) \right)}{\mathsf{r}x}(a) \cdot b = s\cdot \dfrac{\mathsf{r}f(x)}{\mathsf{r}x}(a) \cdot b + t \cdot \dfrac{\mathsf{r}g(x)}{\mathsf{r}x}(a) \cdot b\] 
\item Reverse derivatives are $k$-linear in their second argument:
\[\dfrac{\mathsf{r}f(x)}{\mathsf{r}x}(a) \cdot (s\cdot b + t \cdot c) = s\cdot \dfrac{\mathsf{r}f(x)}{\mathsf{r}x}(a) \cdot b + t \cdot \dfrac{\mathsf{r}f(x)}{\mathsf{r}x}(a) \cdot c\]
\item The reverse derivative of identities and projections are: 
\begin{align*} \dfrac{\mathsf{r}x}{\mathsf{r}x}(a) \cdot b = b && \dfrac{\mathsf{r}x_j}{\mathsf{r}(x_1, \hdots, x_n)}(a_0, \hdots, a_n) \cdot b = \underbrace{(0, \hdots, 0, b, 0, \hdots, 0)}_{\text{$j$th-component}} 
\end{align*} 
\item The reverse derivative of a tuple is the sum of the reverse derivatives of the components of the tuple:
\[\dfrac{\mathsf{r}\left \langle f_1(x), \hdots, f_n(x) \right \rangle}{\mathsf{r}x}(a) \cdot (b_1, \hdots, b_n) = \sum \limits^n_{j=1} \dfrac{\mathsf{r}f_j(x)}{\mathsf{r}x}(a) \cdot b_j\]
\item The reverse chain rule formula for the reverse derivative of a composition: 
\[\dfrac{\mathsf{r}g\left(f(x) \right)}{\mathsf{r}x}(a) \cdot b = \dfrac{\mathsf{r}f(x)}{\mathsf{r}x}(a) \cdot \left( \dfrac{\mathsf{r}g(y)}{\mathsf{r}y}(f(a)) \cdot b \right)\] 
\end{enumerate}

To help write the last two axioms \textbf{[RD.6]} and \textbf{[RD.7]}, it will be useful to first introduce the concept of \emph{partial} reverse derivatives. So, given a map of type $f: A_1 \times \hdots \times A_n \to B$, we'd like to be able to take the reverse derivative of $f$ with respect to the component $A_j$ while keeping the rest constant. To do so, consider first the total reverse derivative of $f$, which is of type $\mathsf{R}[f]: A_1 \times \hdots \times A_n \times B \to A_1 \times \hdots \times A_n$. By the universal property of the product, we know that in this case, $\mathsf{R}[f]$ is actually a tuple of maps of type $A_1 \times \hdots \times A_n \times B \to A_j$ which are, of course, defined by post-composing $\mathsf{R}[f]$ with the respective projections. We interpret these maps as the partial reverse derivatives of $f$.  

\begin{definition}\label{def:partialreverse} In a Cartesian $k$-reverse differential category, for a map $f: A_1 \times \hdots \times A_n \to B$, its \textbf{$j$-th partial reverse derivative} is the map $\mathsf{R}_j[f]: A_1 \times \hdots \times A_n \times B \to A_j$ defined as the following composite:
\begin{align}
 \xymatrixcolsep{5pc}\xymatrix{   \mathsf{R}_j[f]: A_1 \times \hdots \times A_n \times B \ar[rr]^-{\mathsf{R}[f]}  &&  A_1 \times \hdots \times A_n \ar[r]^-{\pi_j} & A_j } 
\end{align}
\end{definition}

In the term calculus, we write partial reverse derivatives as follows: 
\[ \dfrac{\mathsf{r}f(a_1, \hdots, a_{j-1}, x_j, a_{j+1}, \hdots, a_n)}{\mathsf{r}x_j}(a_j) \cdot b ~\colon \!\!  \!\! \!= \pi_j\left(\dfrac{\mathsf{r}f(x_1, \hdots, x_n)}{\mathsf{r}(x_1, \hdots, x_n)}(a_1, \hdots, a_n) \cdot b \right) \]

\begin{example} For a smooth function $F=\langle f_1, \hdots, f_m \rangle: \mathbb{R}^n \to \mathbb{R}^m$, its $j$-th partial reverse derivative $\mathsf{R}_j[F]: \mathbb{R}^n \times \mathbb{R}^m \to \mathbb{R}$ is the sum of the partial derivatives in the $x_j$ variable: 
\begin{align}
\mathsf{R}_j[F](\vec x, \vec y) =  \sum \limits^m_{k=1} \dfrac{\partial f_k}{\partial x_j}(\vec x) y_k
\end{align}
 \end{example}

Recall that the reverse derivative of a smooth function was the tuple of the partial derivatives. The same idea holds in a Cartesian reverse differential category, which is, of course, immediate from the universal property of the product. We note that this is the reverse differential analogue of the fact that the total forward derivative is equal to the sum of the partial forward derivatives \cite[Lemma 2.8.(i)]{garner2020cartesian}. 

\begin{lemma}\label{lemma:tuple} In a Cartesian $k$-reverse differential category, for a map $f: A_1 \times \hdots \times A_n \to B$, we have that $\mathsf{R}[f] = \langle \mathsf{R}_1[f], \hdots, \mathsf{R}_n[f] \rangle$, which in the term calculus is expressed as:
\begin{align} \label{Ri-tuple} 
\dfrac{\mathsf{r}f(x_1, \hdots, x_n)}{\mathsf{r}(x_1, \hdots, x_n)}(a_1, \hdots, a_n) \cdot b =\left \langle  \dfrac{\mathsf{r}f(x_1, a_2, \hdots, a_n)}{\mathsf{r}x_1}(a_1) \cdot b, \hdots,  \dfrac{\mathsf{r}f(a_1, \hdots, x_n)}{\mathsf{r}x_n}(a_n) \cdot b \right \rangle 
\end{align}
\end{lemma}

Now that we have partial reverse derivatives, we can smoothly write down the remaining two axioms: 
\begin{enumerate}[{\bf [RD.1]}]
\setcounter{enumi}{5}
\item Involution of the \emph{linear transpose}:
\[\dfrac{\mathsf{r}\dfrac{\mathsf{r}\dfrac{\mathsf{r}f(x)}{\mathsf{r}x}(a) \cdot u}{\mathsf{r}u}(0) \cdot v}{\mathsf{r}v}(0) \cdot b = \dfrac{\mathsf{r}f(x)}{\mathsf{r}x}(a) \cdot b\]
\item Symmetry of the mixed partial \emph{forward} derivatives: 
\[\dfrac{\mathsf{r} \dfrac{\mathsf{r}\dfrac{\mathsf{r} \dfrac{\mathsf{r}f(x)}{\mathsf{r}x}(y) \cdot u}{\mathsf{r}u}(0) \cdot b}{\mathsf{r}y}(a) \cdot v}{\mathsf{r}v}(0) \cdot c = \dfrac{\mathsf{r} \dfrac{\mathsf{r}\dfrac{\mathsf{r} \dfrac{\mathsf{r}f(x)}{\mathsf{r}x}(y) \cdot u}{\mathsf{r}u}(0) \cdot c}{\mathsf{r}y}(a) \cdot v}{\mathsf{r}v}(0) \cdot b\] 
\end{enumerate}
While these axioms may still look somewhat strange initially, they will make more sense when considering the relationship between reverse differentiation and forward differentiation -- which we discuss in the next section below. For a more in-depth discussion on the seven reverse differential combinator axioms (and Cartesian reverse differential categories in general), we invite the reader to see \cite{cockett2019reverse,cruttwell2022monoidal}. 

Another way of describing partial reverse differentiation is as reverse differentiation in context. From this point of view, partial reverse differentiation is an actual reverse differential combinator for simple slice categories. Recall that for a category $\mathbb{X}$ with finite products, for each object $C \in \mathbb{X}$, the \textbf{simple slice category} over $C$ is the category $\mathbb{X}[C]$ whose objects are the same as $\mathbb{X}$ and whose homsets are $\mathbb{X}[C](A,B) = \mathbb{X}(C \times A, B)$. Composition of $f: C \times A \to B$ and $g: C \times B \to D$ is $g \circ \langle \pi_1, f \rangle$, which is written in the term calculus as $g(c,f(c,x))$, and the identity map is the projection $\pi_2: C \times A \to A$. $\mathbb{X}[C]$ also has finite products, and if $\mathbb{X}$ is a Cartesian left $k$-linear category, then $\mathbb{X}[C]$ will also be a Cartesian $k$-linear category with the same $k$-linear structure as $\mathbb{X}$. Showing that if $\mathbb{X}$ is a Cartesian $k$-reverse differential category then so is $\mathbb{X}[C]$ amounts to showing that partial reverse differentiation satisfies the seven reverse differential combinator axioms in \emph{context}.

\begin{lemma}\label{lem:R17context} In a Cartesian $k$-reverse differential category, the following equalities hold: 
\begin{enumerate}[{\bf [RD.1]}]
\item $\dfrac{\mathsf{r} \left(s \cdot f(c_1,x,c_2) + t \cdot g(c_1,x,c_2)\right)}{\mathsf{r}x}(a) \cdot b = s\cdot \dfrac{\mathsf{r}f(c_1,x,c_2)}{\mathsf{r}x}(a) \cdot b + t \cdot \dfrac{\mathsf{r}g(c_1,x,c_2)}{\mathsf{r}x}(a) \cdot b$
\item $\dfrac{\mathsf{r}f(c_1,x,c_2)}{\mathsf{r}x}(a) \cdot (s\cdot b + t \cdot c) = s\cdot \dfrac{\mathsf{r}f(c_1,x,c_2)}{\mathsf{r}x}(a) \cdot b + t \cdot \dfrac{\mathsf{r}f(c_1,x,c_2)}{\mathsf{r}x}(a) \cdot b$
\item $\dfrac{\mathsf{r}x_j}{\mathsf{r}x_j}(a) \cdot b = b$ and $\dfrac{\mathsf{r}x_j}{\mathsf{r}x_i}(a) \cdot b = 0$ if $i \neq j$; 
\item $\dfrac{\mathsf{r}\left \langle f_1(c_1,x,c_2), \hdots, f_n(c_1,x,c_2) \right \rangle}{\mathsf{r}x}(a) \cdot (b_0, \hdots, b_n) = \sum \limits^n_{j=1} \dfrac{\mathsf{r}f_j(c_1,x,c_2)}{\mathsf{r}x}(a) \cdot b_j$
\item $\dfrac{\mathsf{r}g\left(c_1, f(c_1,x,c_2), c_2\right)}{\mathsf{r}x}(a) \cdot b = \dfrac{\mathsf{r}f(c_1,x,c_2)}{\mathsf{r}x}(a) \cdot \left( \dfrac{\mathsf{r}g(c_1,y,c_2)}{\mathsf{r}y}(f(c_1,a,c_2)) \cdot b \right)$
\item $\dfrac{\mathsf{r}\dfrac{\mathsf{r}\dfrac{\mathsf{r}f(c_1,x,c_2)}{\mathsf{r}x}(a) \cdot u}{\mathsf{r}u}(0) \cdot v}{\mathsf{r}v}(0) \cdot b = \dfrac{\mathsf{r}f(c_1,x,c_2)}{\mathsf{r}x}(a) \cdot b$
\item $\dfrac{\mathsf{r} \dfrac{\mathsf{r}\dfrac{\mathsf{r} \dfrac{\mathsf{r}f(c_1,x,c_2)}{\mathsf{r}x}(y) \cdot u}{\mathsf{r}u}(0) \cdot b}{\mathsf{r}y}(a) \cdot v}{\mathsf{r}v}(0) \cdot c = \dfrac{\mathsf{r} \dfrac{\mathsf{r}\dfrac{\mathsf{r} \dfrac{\mathsf{r}f(c_1,x,c_2)}{\mathsf{r}x}(y) \cdot u}{\mathsf{r}u}(0) \cdot c}{\mathsf{r}y}(a) \cdot v}{\mathsf{r}v}(0) \cdot b$
\end{enumerate}
\end{lemma}
\begin{proof} To prove these, we first need to compute the following useful identity: 
\begin{gather*}
\dfrac{\mathsf{r}\left\langle c_1, f(c_1,x,c_2), c_2 \right \rangle}{\mathsf{r}x}(a) \cdot (b_1, b_2, b_3) ~~\substack{= \\ \text{Def.}}~~  \pi_2\left( \dfrac{\mathsf{r}\left\langle v, f(v,x,u), u \right \rangle}{\mathsf{r}(v,x,u)}(c_1,a,c_2) \cdot (b_1, b_2, b_3) \right) \\
    \substack{= \\ \text{\textbf{[R.4]}}}~~  \pi_2\left(\dfrac{\mathsf{r}v}{\mathsf{r}(v,x,u)}(c_1,a,c_2) \cdot b_1 +  \dfrac{\mathsf{r} f(v,x,u)}{\mathsf{r}(v,x,u)}(c_1,a,c_2) \cdot b_2 + \dfrac{\mathsf{r}u}{\mathsf{r}(v,x,u)}(c_1,a,c_2) \cdot b_3 \right)  \\ 
  \substack{= \\ \text{\textbf{[R.3]} + (\ref{Ri-tuple})}}~~ \pi_2\left( (b_1,0,0) + \left ( \dfrac{\mathsf{r} f(v,a,c_2)}{\mathsf{r}v}(c_1) \cdot b_2, \dfrac{\mathsf{r} f(c_1,x,c_2)}{\mathsf{r}x}(a) \cdot b_2, \dfrac{\mathsf{r} f(c_1,a,u)}{\mathsf{r}u}(c_2) \cdot b_2  \right ) + (0,0,b_3) \right) \\
\substack{= \\ \text{$\pi_j$ is $k$-lin.}}~~  \pi_2(b_1,0,0) + \pi_2\left(  \dfrac{\mathsf{r} f(v,a,c_2)}{\mathsf{r}v}(c_1) \cdot b_2, \dfrac{\mathsf{r} f(c_1,x,c_2)}{\mathsf{r}x}(a) \cdot b_2, \dfrac{\mathsf{r} f(c_1,a,u)}{\mathsf{r}u}(c_2) \cdot b_2  \right ) + \pi_2(0,0,b_3) \\ 
 = 0 + \dfrac{\mathsf{r} f(c_1,x,c_2)}{\mathsf{r}x}(a) \cdot b_2 + 0 = \dfrac{\mathsf{r} f(c_1,x,c_2)}{\mathsf{r}x}(a) \cdot b_2 
\end{gather*} 
So we have that: 
\begin{align}\label{useful1}
    \dfrac{\mathsf{r}\left\langle c_1, f(c_1,x,c_2), c_2 \right \rangle}{\mathsf{r}x}(a) \cdot (b_1, b_2, b_3) = \dfrac{\mathsf{r} f(c_1,x,c_2)}{\mathsf{r}x}(a) \cdot b_2
\end{align}
Using this and the corresponding axiom \textbf{[RD.\#]}, we can prove the corresponding partial reverse version. Let's prove the reverse chain rule in context: 
\begin{gather*}
  \dfrac{\mathsf{r}g\left(c_1, f(c_1,x,c_2), c_2\right)}{\mathsf{r}x}(a) \cdot b ~~\substack{= \\ \text{Def.}}~~ \pi_2\left( \dfrac{\mathsf{r}g\left(v, f(v,x,u), u\right)}{\mathsf{r}(v,x,u)}(c_1,a,c_2) \cdot b  \right) \\
     \substack{= \\ \text{\textbf{[R.5]}}}~~  \pi_2\left( \dfrac{\mathsf{r}\left\langle v, f(v,x,u), u \right \rangle}{\mathsf{r}(v,x,u)}(c_1,a,c_2) \cdot \left( \dfrac{\mathsf{r}g(w,y,z)}{\mathsf{r}(w,y,z)}(c_1, f(c_1,a,c_2), c_2) \cdot b \right) \right) \\ 
\substack{= \\ \text{Def.}}~~ \dfrac{\mathsf{r}\left\langle c_1, f(c_1,x,c_2), c_2 \right \rangle}{\mathsf{r}x}(a) \cdot \left( \dfrac{\mathsf{r}g(w,y,z)}{\mathsf{r}(w,y,z)}(c_1, f(c_1,a,c_2), c_2) \cdot b \right) \\ 
 \substack{= \\ \text{(\ref{Ri-tuple})}}~~ \text{\footnotesize $\dfrac{\mathsf{r}\left\langle c_1, f(c_1,x,c_2), c_2 \right \rangle}{\mathsf{r}x}(a) \cdot \left( \dfrac{\mathsf{r}g\left(w, f(c_1,a,c_2), c_2 \right)}{\mathsf{r}w}(c_1) \cdot b, \dfrac{\mathsf{r}g(c_1,y,c_2)}{\mathsf{r}y}(f(c_1,a,c_2)) \cdot b,  \dfrac{\mathsf{r}g(c_1,f(c_1,a,c_2),z)}{\mathsf{r}z}(c_2) \cdot b \right)$ } \\ 
  \substack{= \\ \text{(\ref{useful1})}}~~  \dfrac{\mathsf{r} f(c_1,x,c_2)}{\mathsf{r}x}(a) \cdot \dfrac{\mathsf{r}g(c_1,y,c_2)}{\mathsf{r}y}(f(c_1,a,c_2)) \cdot b
\end{gather*} 
The other identities are shown in a similar fashion. 
\end{proof}

Then, it follows from the above lemma that the simple slice categories of a Cartesian reverse differential category are again Cartesian reverse differential categories. This is the reverse differentiation analogue of \cite[Cor 4.5.2]{blute2009cartesian}. 

\begin{corollary} Let $\mathbb{X}$ be a Cartesian $k$-reverse differential category. Then for each object $C \in \mathbb{X}$, the simple slice category $\mathbb{X}[C]$ is a Cartesian $k$-reverse differential category whose reverse differential combinator $\mathsf{R}^C$ sends a map $f: C \times A \to B$ to the map $\mathsf{R}^C[f]: C \times A \times B \to A$ defined as the following composite: 
\begin{align}
    \xymatrixcolsep{5pc}\xymatrix{   \mathsf{R}^C[f]: C \times A \times B \ar[r]^-{\mathsf{R}[f]}  &  C \times A \ar[r]^-{\pi_2} & A } 
\end{align} 
\end{corollary}
\begin{proof} The necessary reverse differential combinator axioms with respect to composition in the simple slice category are precisely the identities from Lemma \ref{lem:R17context} with context only on the left. 
\end{proof}

\section{Forward Differentiation and Linear Transpose}\label{sec:forward}

The fundamental theorem for Cartesian reverse differential categories is that they are precisely the same as Cartesian \emph{forward} differential categories equipped with a suitable notion of transpose operation, called a \emph{contextual linear dagger}. In particular, this says we obtain forward differentiation and a transpose operation from reverse differentiation, and vice-versa. 

\begin{theorem} \cite[Thm 42]{cockett2019reverse} A Cartesian $k$-reverse differential category is precisely a Cartesian $k$-differ-ential category with a contextual $k$-linear dagger.  
\end{theorem}
\begin{proof} The proof in the relative to $k$ setting is essentially the same as was done in \cite{cockett2019reverse}. 
\end{proof}

We will not thoroughly review Cartesian differential categories and contextual linear daggers here. For a more in-depth introduction to Cartesian forward differential categories, we invite the reader to see \cite{blute2009cartesian,garner2020cartesian,Lemaygradient}, and for an in-depth introduction to contextual linear daggers, we invite the reader to see \cite{cockett2019reverse,cruttwell2022monoidal}. Instead, we will focus on revisiting how to build a forward differential combinator and contextual linear dagger from a reverse differential combinator. We will also see how these constructions can be nicely expressed using partial reverse derivatives. 

Briefly, a \textbf{Cartesian $k$-differential category} \cite[Sec 2.2]{garner2020cartesian} is a Cartesian $k$-linear category which comes equipped with a \textbf{forward differential combinator}, which is an operator which sends maps ${f: A \to B}$ to their \textbf{forward derivative} $\mathsf{D}[f]: A \times A \to B$. The differential combinator axioms are analogues of the well-known identities of the total derivative from differential calculus, such as the chain rule. As previously mentioned, Cartesian differential categories have a very practical term calculus \cite[Sec 4]{blute2009cartesian}. So, we write the forward derivative as:
\[\mathsf{D}[f](a,b) \colon = \dfrac{\mathsf{d}f(x)}{\mathsf{d}x}(a) \cdot b\] 
In particular, one the axioms for the forward differential combinator is \textbf{[CD.5]} the chain rule for differentiating a composition of maps, which in the term calculus is expressed as: 
\begin{align}
    \dfrac{\mathsf{d}g\left(f(x) \right)}{\mathsf{d}x}(a) \cdot b = \dfrac{\mathsf{d}g(y)}{\mathsf{d}y}(f(a)) \cdot \left( \dfrac{\mathsf{d}f(x)}{\mathsf{d}x}(a) \cdot b \right)
\end{align}

Let's now revisit how to go from reverse differentiation to forward differentiation by rewriting the construction of \cite[Thm 16]{cockett2019reverse} using partial reverse derivatives. Consider a map $f: A \to B$. We need to build a map of type $A \times A \to B$. We can first take its reverse derivative to get $\mathsf{R}[f]: A \times B \to A$, and then by taking the partial reverse derivative with respect to $B$, we get a map of type $\mathsf{R}_2\left[\mathsf{R}[f] \right]: A \times B \times A \to B$. Then inserting zero into the $B$ argument of the domain gives us a map of type $A \times A \to B$ as desired. Explicitly, for a map $f: A \to B$, its forward derivative is defined as follows: 
\begin{align}\label{RD}
     \dfrac{\mathsf{d}f(x)}{\mathsf{d}x}(a) \cdot b \colon =   \dfrac{\mathsf{r} \dfrac{\mathsf{r}f(x)}{\mathsf{r}x}(a) \cdot y}{\mathsf{r}y}(0) \cdot b 
\end{align}
Then \textbf{[RD.7]} is precisely expressing the axiom \textbf{[CD.7]} of the forward differential combinator (which amounts to the symmetry of the mixed partial forward derivatives) using only reverse differentiation. 

\begin{example}\label{ex:smoothCDC} In $\mathsf{SMOOTH}$, the forward differential combinator is the usual total derivative from calculus, which is instead defined using the Jacobian \cite[Ex 5.3]{cockett2019reverse}. Explicitly, for a smooth function $F: \mathbb{R}^n \to \mathbb{R}^m$ (where recall that $F$ is in fact a tuple $F = \langle f_1, \hdots, f_m \rangle$ of smooth functions $f_j:  \mathbb{R}^n \to \mathbb{R}$), its forward derivative $\mathsf{D}[F]: \mathbb{R}^n \times \mathbb{R}^n \to \mathbb{R}^m$ is defined as follows: 
\begin{align*}
     \mathsf{D}[F](\vec x, \vec y) := \left(\sum \limits^n_{i=1} \frac{\partial f_1}{\partial x_i}(\vec x) y_i, \hdots, \sum \limits^n_{i=1} \frac{\partial f_m}{\partial x_i}(\vec x) y_i \right) 
\end{align*}
\end{example}

Before defining the transpose operation, we first need to quickly discuss \emph{partial} forward derivatives and \emph{differential linearity}. Starting with partial forward derivatives: given a map $f: A_0 \times \hdots \times A_n \to B$, we'd like to take the forward derivative of $f$ with respect to the component $A_j$ while keeping the rest constant. In differential calculus, partial derivatives are obtained by inserting zeroes in the appropriate vector argument of the total derivative. The same idea holds in a Cartesian differential category. Then the \textbf{$j$-th partial forward derivative} \cite[Def 2.7]{garner2020cartesian} of $f$ is the map $\mathsf{D}_j[f]:  A_1 \times \hdots \times A_n \times A_j \to B$ written in term calculus on the left below and is defined on the right as follows: 
\begin{align}\dfrac{\mathsf{d}f(a_1, \hdots, a_{j-1}, x_j, a_{j+1}, \hdots, a_n)}{\mathsf{d}x_j}(a_j) \cdot b ~~\colon \!\!\!\!= \dfrac{\mathsf{d}f(x_1, \hdots, x_n)}{\mathsf{d}(x_1, \hdots, x_n)}(a_1, \hdots, a_n) \cdot (0, \hdots, 0, b, 0, \hdots, 0) 
\end{align}
A map $f: A_1 \times \hdots \times A_n \to B$ is said to be \textbf{differential linear ($\mathsf{D}$-linear)} \cite[Def 2.6]{garner2020cartesian} in $A_j$ if when taking its $j$-th partial derivative, the following equality holds: 
\[ \dfrac{\mathsf{d}f(a_1, \hdots, a_{j-1}, x_j, a_{j+1}, \hdots, a_n)}{\mathsf{d}x_j}(a_j) \cdot b = f(a_1, \hdots, a_{j-1}, b, a_{j+1}, \dots, a_n) \]
Properties of differential linear maps can be found in \cite[Lemma 2.6]{cockett2020linearizing}, such as the fact that they are closed under composition, $k$-linear structure, etc. In particular, if $f$ is $\mathsf{D}$-linear in its $j$-th variable, then it is also $k$-linear in its $j$-th variable \cite[Lemma 2.6.(i)]{cockett2020linearizing} (though the converse is not necessarily true in an arbitrary Cartesian differential category). Moreover, the axiom \textbf{[CD.6]} of the forward differential combinator is precisely the statement that forward derivatives are $\mathsf{D}$-linear in their second argument.

A \textbf{contextual linear dagger} \cite[Def 39]{cockett2019reverse} is an involutive and contravariant operation on maps with a $\mathsf{D}$-linear argument, which swaps the codomain with said $\mathsf{D}$-linear argument. Here we will write down the contextual linear dagger with context both on the left and right, which is equivalent to simply having context on the left as in \cite{cockett2019reverse}. So for a map $f: C_1 \times A \times C_2 \to B$ which is $\mathsf{D}$-linear in $A$, the \textbf{$\mathsf{D}$-linear transpose in $A$} of $f$ is the map the map $f^{\dagger[C_1 \times \_ \times C_2]}: C_1 \times B \times C_2 \to A$ defined as the partial reverse derivative in $A$ of $f$ evaluated at $0$:  
\begin{align}
     f^{\dagger[C_1 \times \_ \times C_2]}(c_1,b,c_2) \colon \!\!\!=   \dfrac{\mathsf{r} f(c_1, x, c_2)}{\mathsf{r}x}(0) \cdot b 
\end{align}
Now $f^{\dagger[C_1 \times \_ \times C_2]}: C_1 \times B \times C_2 \to A$ is $\mathsf{D}$-linear in $B$ \cite[Cor 36]{cockett2019reverse} and this operation is involutive \cite[Lemma 35]{cockett2019reverse}, that is, the $\mathsf{D}$-linear transpose of $f^{\dagger[C_1 \times \_ \times C_2]}$ is $f$. Other properties of the contextual linear dagger can be found in \cite{cockett2019reverse,cruttwell2022monoidal}. In particular, it will be useful for the reverse Faà di Bruno's Formula to recall the fact that the contextual linear dagger is contravariant with respect to composition in context. 

\begin{lemma} In a Cartesian $k$-reverse differential category, if $f: C_1 \times A \times C_2 \to B$ is $\mathsf{D}$-linear in $A$ and $g: C_1 \times B \times C_2 \to E$ is $\mathsf{D}$-linear in $B$, then $(g \circ \langle \pi_1, f, \pi_3 \rangle)^{\dagger[C_1 \times \_ \times C_2]} = f^{\dagger[C_1 \times \_ \times C_2]} \circ \langle \pi_1, g^{\dagger[C_1 \times \_ \times C_2]}, \pi_3 \rangle$, which in the term logic is explicitly expressed as: 
\begin{align}\label{eq:contra}
\dfrac{\mathsf{r} g\left(c_1, f(c_1, x, c_2), c_2\right)}{\mathsf{r}x} (0) \cdot b = \dfrac{\mathsf{r} f(c_1, x, c_2)}{\mathsf{r}x} (0) \cdot \left( \dfrac{\mathsf{r} g\left(c_1, y, c_2\right)}{\mathsf{r}y} (0) \cdot b \right) 
\end{align}
\end{lemma}
\begin{proof} The proof is essentially the same as the equational calculations done in the proof of \cite[Thm 37]{cockett2019reverse}. However, since it is crucial to proving the reverse Faà di Bruno's Formula, revisiting the proof using the term calculus may be useful. So first recall that since $f$ is $\mathsf{D}$-linear in $A$, it is also $k$-linear in $A$. So in particular $f(c_1, 0, c_2) =0$. So then we compute: 
\begin{gather*}
    \dfrac{\mathsf{r} g\left(c_1, f(c_1, x, c_2), c_2\right)}{\mathsf{r}x} (0) \cdot b  ~~\substack{= \\ \text{\textbf{[R.5]}}}~~ \dfrac{\mathsf{r} f(c_1, x, c_2)}{\mathsf{r}x} (0) \cdot \left( \dfrac{\mathsf{r} g\left(c_1, y, c_2\right)}{\mathsf{r}y} \left( f(c_1, 0, c_2)  \right) \cdot b \right) \\
    \substack{= \\ \text{$\mathsf{D}$-lin. $\Rightarrow$ $k$-lin.}} \dfrac{\mathsf{r} f(c_1, x, c_2)}{\mathsf{r}x} (0) \cdot \left( \dfrac{\mathsf{r} g\left(c_1, y, c_2\right)}{\mathsf{r}y} (0) \cdot b \right) 
\end{gather*}
So the desired equality holds. \end{proof}

Now \cite[Thm 42]{cockett2019reverse} tells us that the $\mathsf{D}$-linear transpose of the forward derivative is the reverse derivative, that is, $\mathsf{D}[f]^{\dagger[A \times \_]} = \mathsf{R}[f]$, which in term calculus is written as: 
\begin{align}\label{eq:11}
    \dfrac{\mathsf{r}\dfrac{\mathsf{d}f(x)}{\mathsf{d}x}(a) \cdot u}{\mathsf{r}u}(0) \cdot b = \dfrac{\mathsf{r}f(x)}{\mathsf{r}x}(a) \cdot b 
 \end{align}
which we note is precisely \textbf{[RD.6]}. Moreover, this also implies that $\mathsf{R}[f]$ is $\mathsf{D}$-linear in its second argument and that the $\mathsf{D}$-linear transpose of the reverse derivative is the forward derivative, that is, $\mathsf{R}[f]^{\dagger[A \times \_]} = \mathsf{D}[f]$. Furthermore, we can also show that the partial reverse derivative in the second argument of the reverse derivative is precisely the forward derivative: 

\begin{lemma}\label{lemma:profR}In a Cartesian $k$-reverse differential category, the following equalities hold: 
\begin{align} 
    \dfrac{\mathsf{r}\dfrac{\mathsf{r}f(x)}{\mathsf{r}x}(a_1) \cdot u}{\mathsf{r}u}(b) \cdot a_2 = \dfrac{\mathsf{d}f(x)}{\mathsf{d}x}(a_1) \cdot a_2 
\end{align}
\end{lemma}
\begin{proof} We compute: 
    \begin{gather*}
        \dfrac{\mathsf{r}\dfrac{\mathsf{r}f(x)}{\mathsf{r}x}(a_1) \cdot u}{\mathsf{r}u}(b) \cdot a_2  ~~\substack{= \\ \text{(\ref{eq:11})}}~~ \dfrac{\mathsf{r}  \dfrac{\mathsf{d}\dfrac{\mathsf{r}f(x)}{\mathsf{r}x}(a_1) \cdot u}{\mathsf{d}u}(b) \cdot v}{\mathsf{r}v}(0) \cdot a_2 ~~\substack{= \\ \text{$\mathsf{D}$-lin.}}~~  \dfrac{\mathsf{r} \dfrac{\mathsf{r}f(x)}{\mathsf{r}x}(a_1) \cdot v}{\mathsf{r}v}(0) \cdot a_2 ~~\substack{= \\ \text{(\ref{RD})}}~~   \dfrac{\mathsf{d}f(x)}{\mathsf{d}x}(a_1) \cdot a_2
    \end{gather*}
So the desired equality holds. \end{proof}

We conclude this section by showing that the partial reverse derivative is the $\mathsf{D}$-linear transpose of the partial forward derivative. 

\begin{lemma}\label{lemma:partialtranspose} In a Cartesian $k$-reverse differential category, for every map $f: A_1 \times \hdots \times A_n \to B$, we have that $\mathsf{D}_j[f]^{\dagger[A_1 \times \hdots \times A_n \times \_]} = \mathsf{R}_j[f]$, which in the term calculus is expressed as follows: 
\begin{align}
    \dfrac{\mathsf{r}\dfrac{\mathsf{d}f(a_1, \hdots, a_{j-1}, x_j, a_{j+1}, \hdots, a_n)}{\mathsf{d}x_j}(a_j) \cdot u}{\mathsf{r}u}(0) \cdot b = \dfrac{\mathsf{r}f(a_1, \hdots, a_{j-1}, x_j, a_{j+1}, \hdots, a_n)}{\mathsf{r}x_j}(a_j) \cdot b
\end{align}
\end{lemma}
\begin{proof} Without loss of generality, for simplicity, we prove this for a map of type $f: C_1 \times A \times C_2 \to B$. So we compute that: 
\begin{gather*}
     \dfrac{\mathsf{r}\dfrac{\mathsf{d}f(c_1, x, c_2)}{\mathsf{d}x}(a) \cdot u}{\mathsf{r}u}(0) \cdot b ~~\substack{= \\ \text{Def.}}~~  \dfrac{\mathsf{r}\dfrac{\mathsf{d}f(v, x, w)}{\mathsf{d}(v,x,w)}(c_1,a,c_2) \cdot (0,u,0)}{\mathsf{r}u}(0) \cdot b  \\
     \substack{= \\ \text{(\ref{RD})}}~~ \text{\small $\dfrac{\mathsf{r}  \dfrac{\mathsf{r} \dfrac{\mathsf{r} f(v,x,w)}{\mathsf{r}(v,x,w)}(c_1,a,c_2) \cdot y  }{\mathsf{r} y}(0) \cdot (0,u,0)   }{\mathsf{r}u}(0) \cdot b ~~\substack{= \\ \text{Def.}}~~ \pi_2 \left( \dfrac{\mathsf{r}  \dfrac{\mathsf{r} \dfrac{\mathsf{r} f(v,x,w)}{\mathsf{r}(v,x,w)}(c_1,a,c_2) \cdot y  }{\mathsf{r} y}(0) \cdot (s,u,t)   }{\mathsf{r}(s,u,t)}(0,0,0) \cdot b \right)$} \\
     \substack{= \\ \text{\textbf{[R.6]}}}~~ \pi_2\left( \dfrac{\mathsf{r} f(v,x,w)}{\mathsf{r}(v,x,w)}(c_1,a,c_2) \cdot b \right) ~~\substack{= \\ \text{Def.}}~~ \dfrac{\mathsf{r} f(c_1,x,c_2)}{\mathsf{r}x}(a) \cdot b
\end{gather*}
So the desired equality holds. \end{proof}

\section{Reverse Faà di Bruno's Formula}

To provide a reverse differentiation version of Faà di Bruno's Formula, we must first work out the appropriate notion of \emph{higher-order reverse derivatives}. Now first observe that given a map $f: A \to B$, reverse deriving it once gives $\mathsf{R}[f]: A \times B \to A$, reverse deriving it again gives $\mathsf{R}^2[f]: A \times B \times A \to A \times B$, then reverse deriving again gives $\mathsf{R}^3[f]: A \times B \times A \times A \times B \to A \times B \times A$, so on and so forth. We quickly see that both the domain and codomain of $\mathsf{R}^n[f]$ expand quite rapidly, and thus $\mathsf{R}^n[f]$ is not necessarily easy to work with. However, it turns out that from Lemma \ref{lemma:tuple} and Lemma \ref{lemma:profR}, there is a lot of redundant information in $\mathsf{R}^n[f]$. For example, we can compute that $\mathsf{R}^2[f]$ is worked out to be: 
\[ \dfrac{\mathsf{r}\dfrac{\mathsf{r}f(x)}{\mathsf{r}x}(y) \cdot z}{\mathsf{r}(y,z)}(a_1,b) \cdot a_2  = \left\langle  \dfrac{\mathsf{r}\dfrac{\mathsf{r}f(x)}{\mathsf{r}x}(a_1) \cdot u}{\mathsf{r}u}(b) \cdot a_2, \dfrac{\mathsf{d}f(x)}{\mathsf{d}x}(a_1) \cdot a_2 \right \rangle \]
So we see that $\mathsf{R}^2[f]$ contains previous known information, since $\mathsf{D}[f]$ is the $\mathsf{D}$-linear transpose of $\mathsf{R}[f]$. So, all the new information comes from taking the partial reverse derivative in the first argument. This is how we get the higher-order reverse derivatives. So starting with a map $f: A \to B$, we again first take its reverse derivative $\mathsf{R}[f]: A \times B \to A$, and then this time only take the partial reverse derivative in the first argument to get $\mathsf{R}_1\left[ \mathsf{R}[f] \right]: A \times B \times A \to A$, then do this again to get $\mathsf{R}_1\left[\mathsf{R}_1\left[ \mathsf{R}[f] \right]\right]: A \times B \times A \times A \to A$, and so on. So after doing this $n+1$ times, we get a map of type $A \times B \times A^{\times^n} \to A$. 

\begin{definition}\label{def:higherreverse} In a Cartesian $k$-reverse differential category, for a map $f: A \to B$, the \textbf{$n+1$-th reverse derivative} of $f$ is the map $\rho^{(n+1)}[f]: A \times B \times A^{\times^n} \to A$ defined inductively as $\rho^{(1)}[f] = \mathsf{R}[f]$ and $\rho^{(n+2)}[f] = \mathsf{R}_1\left[ \rho^{(n+1)}[f] \right]$. 
\end{definition}

In the term calculus, inspired by the notation used in \cite[Sec 3.1]{cockett2011faa} for higher-order forward derivatives, we write higher-order reverse derivatives as follows: 
\[ \rho^{(n+1)}[f](a_0, b, a_2, \hdots, a_{n+1}) = \dfrac{\mathsf{r}^{(n+1)} f(x)}{\mathsf{r}x}(a_0) \cdot b \cdot a_2 \cdot \hdots \cdot a_{n+1} \]
As such, by definition, we have that: 
\begin{align*}
    \dfrac{\mathsf{r}^{(1)} f(x)}{\mathsf{r}x}(a_0) \cdot b &=  \dfrac{\mathsf{r} f(x)}{\mathsf{r}x}(a_0) \cdot b\\
    \dfrac{\mathsf{r}^{(n+2)} f(x)}{\mathsf{r}x}(a_0) \cdot b \cdot a_2 \cdot \hdots \cdot a_{n+1} &= \dfrac{\mathsf{r}\dfrac{\mathsf{r}^{(n+1)} f(x)}{\mathsf{r}x}(y) \cdot b \cdot a_2 \cdot \hdots \cdot a_{n+1}}{\mathsf{r}y}(a_0) \cdot a_{n+2}
\end{align*}

Higher-order forward derivatives are defined similarly. Starting again with a map $f: A \to B$, we can repeatedly derive the first argument to get a map of type $A \times A^{\times^{n+1}} \to B$. So the \textbf{$n+1$-th forward derivative} \cite[Def 3.1]{garner2020cartesian} of $f$ is the map $\partial^{(n+1)}[f]: A \times A^{\times^n} \to B$, which is defined inductively as $\partial^{(1)}[f] = \mathsf{D}[f]$ and $\partial^{(n+2)}[f] = \mathsf{D}_1\left[ \partial^{(n+1)}[f] \right]$. In the term calculus, following the notation used in \cite[Sec 3.1]{cockett2011faa}, we write higher-order forward derivatives as: 
\begin{align*}
   \partial^{(n+1)}[f](a_0, a_1, \hdots, a_{n+1}) \colon = \dfrac{\mathsf{d}^{(n+1)} f(x)}{\mathsf{d}x}(a_0) \cdot a_1 \cdot \hdots \cdot a_{n+1} 
\end{align*}
So, by definition, we have that:
\begin{align*}
    \dfrac{\mathsf{d}^{(1)} f(x)}{\mathsf{d}x}(a_0) \cdot b &=  \dfrac{\mathsf{d} f(x)}{\mathsf{d}x}(a_0) \cdot b\\
    \dfrac{\mathsf{d}^{(n+2)} f(x)}{\mathsf{d}x}(a_0) \cdot a_1 \cdot a_2 \cdot \hdots \cdot a_{n+1} &= \dfrac{\mathsf{d}\dfrac{\mathsf{d}^{(n+1)} f(x)}{\mathsf{d}x}(y) \cdot a_1 \cdot a_2 \cdot \hdots \cdot a_{n+1}}{\mathsf{d}y}(a_0) \cdot a_{n+2}
\end{align*}

Now the $n+1$-th forward derivative is $\mathsf{D}$-linear in each of the last $n+1$ arguments $A$ and is also symmetric in its last $n+1$ arguments \cite[Lemma 3.1.(i)]{garner2020cartesian}. As such, it does not matter which $\mathsf{D}$-linear argument we transpose. So for convenience, to line up with the type higher-order reverse derivative, we will consider the $\mathsf{D}$-linear transpose of the second argument of the higher-order forward derivative: $\partial^{(n+1)}[f]^{\dagger[A \times \_ \times A^n]}: A \times B \times A^{\times^n} \to A$. Now we'd like that $\partial^{(n+1)}[f]^{\dagger[A \times \_ \times A^n]}$ be equal to $\rho^{(n+1)}$. Unfortunately, this does not seem to necessarily follow from just \textbf{[RD.1]}-\textbf{[RD.7]}. As such, we introduce an extra compatibility between the forward differential combinator and the reverse differential combinator. 

Given a map $f: A \to B$, we can take its forward derivative $\mathsf{D}[f]: A \times A \to B$, and then take its reverse derivative in the first argument to get $\mathsf{R}_1\left[ \mathsf{D}[f] \right]: A \times A \times B \to A$. On the other hand, we could have first taken the reverse derivative $\mathsf{R}[f]: A \times B \to A$, and then taken its reverse derivative in the first argument to get $\mathsf{R}_1\left[ \mathsf{R}[f] \right]: A \times B \times A \to A$. Up to swapping the last two arguments, we ask that these be equal. 

\begin{definition}\label{def:stable} A Cartesian $k$-reverse differential category is said to satisfy the \textbf{stable rule} if: 
\begin{align}
    \dfrac{\mathsf{r} \dfrac{\mathsf{d} f(x)}{\mathsf{d}x}(y) \cdot a_2 }{\mathsf{r} y} (a_1) \cdot b = \dfrac{\mathsf{r} \dfrac{\mathsf{r} f(x)}{\mathsf{r}x}(y) \cdot b }{\mathsf{r} y} (a_1) \cdot a_2
\end{align}
\end{definition}

To the best of our understanding at this moment, the stable rule does not seem to follow from the reverse differential combinator axioms -- though we do not have a proper counter-example at present. The stable rule is a desirable coherence that should hold, and all natural examples of Cartesian reverse differential categories satisfy the stable rule. In future work, it may be worth revisiting the axioms of the reverse differential combinator to include the stable rule, which we conjecture is related to taking multiple linear transposes of multi-$\mathsf{D}$-linear maps.  

\begin{example}\label{ex:stablesmooth} Let us show that smooth functions satisfy the stable rule. For a smooth function $F= \langle f_1, \hdots, f_m \rangle: \mathbb{R}^n \to \mathbb{R}^m$, recall that $\mathsf{R}[F]: \mathbb{R}^n \times \mathbb{R}^m \to \mathbb{R}^n$ and ${\mathsf{D}[F]: \mathbb{R}^n \times \mathbb{R}^n \to \mathbb{R}^m}$ are: 
\[ \mathsf{R}[F](\vec x, \vec y) = \left( \mathsf{R}_1[F](\vec x, \vec y), \hdots, \mathsf{R}_n[F](\vec x, \vec y) \right) = \left(  \sum \limits^m_{j=1} \frac{\partial f_j}{\partial x_1}(\vec x) y_j, \hdots, \sum \limits^m_{j=1} \frac{\partial f_j}{\partial x_n}(\vec x) y_j \right)  \]
\[ \mathsf{D}[F](\vec x, \vec z) = \left( \mathsf{D}[f_1](\vec x, \vec z), \hdots,  \mathsf{D}[f_m](\vec x, \vec z) \right) = \left( \sum\limits^n_{i=1} \dfrac{\partial f_1}{\partial x_i}(\vec x) z_i, \hdots, \sum\limits^n_{i=1} \dfrac{\partial f_m}{\partial x_i}(\vec x) z_i \right) \]
On the other hand, for a smooth function $G = \langle g_1, \hdots, g_k \rangle: \mathbb{R}^n \times \mathbb{R}^m \to \mathbb{R}^k$, its partial reverse derivative in its first argument $\mathbb{R}^n$ is $\mathsf{R}_1[G]: \mathbb{R}^n \times \mathbb{R}^m \times \mathbb{R}^k \to \mathbb{R}^n$ given as follows:
\begin{align*}
  \mathsf{R}_1[G](\vec x, \vec y, \vec z) = \left (  \sum \limits^k_{j=1} \dfrac{\partial g_j}{\partial x_1}(\vec x, \vec y) z_j, \hdots,\sum \limits^k_{j=1} \dfrac{\partial g_j}{\partial x_n}(\vec x,\vec y) z_j \right )  
\end{align*}
We need to show that $\mathsf{R}_1\left[ \mathsf{D}[F] \right]: \mathbb{R}^n \times \mathbb{R}^n \times \mathbb{R}^m \to \mathbb{R}^n$ and $\mathsf{R}_1\left[ \mathsf{R}[F] \right]: \mathbb{R}^n \times \mathbb{R}^m \times \mathbb{R}^n \to \mathbb{R}^n$ are equal up to permutation of the second and third arguments. We first compute $\mathsf{R}_1\left[ \mathsf{D}[F] \right]$ to be:
\begin{align*}
    \mathsf{R}_1\left[ \mathsf{D}[F] \right](\vec x, \vec z, \vec y) &=~ \left (  \sum \limits^m_{j=1} \dfrac{\partial \mathsf{D}[f_j]}{\partial x_1}(\vec x, \vec z) y_j, \hdots,\sum \limits^m_{j=1} \dfrac{\partial \mathsf{D}[f_j]}{\partial x_n}(\vec x,\vec z) y_j \right ) \\
    &=~ \left ( \sum \limits^m_{j=1} \sum \limits^n_{i=1} \dfrac{\partial^2 f_j}{\partial x_i \partial x_1}(\vec x) z_i y_j, \hdots,\sum \limits^m_{j=1} \sum \limits^n_{i=1} \dfrac{\partial^2 f_j}{\partial x_i \partial x_n}(\vec x) z_i y_j \right)
\end{align*}
On the other hand, we compute $\mathsf{R}_1\left[ \mathsf{R}[F] \right]$ to be:
\begin{align*}
    \mathsf{R}_1\left[ \mathsf{R}[F] \right](\vec x, \vec y, \vec z) &=~ \left(  \sum \limits^n_{i=1} \dfrac{\partial \mathsf{R}_i[F]}{\partial x_1}(\vec x) z_i, \hdots, \sum \limits^n_{i=1} \dfrac{\partial \mathsf{R}_i[F]}{\partial x_n}(\vec x) z_i \right) \\
    &=~ \left (  \sum \limits^n_{i=1} \sum \limits^m_{j=1} \dfrac{\partial^2 f_j}{\partial x_i \partial x_1}(\vec x) y_j z_i, \hdots, \sum \limits^n_{i=1} \sum \limits^m_{j=1} \dfrac{\partial^2 f_j}{\partial x_i \partial x_n}(\vec x) y_j z_i \right) 
\end{align*}
So we have that $\mathsf{R}_1\left[ \mathsf{D}[F] \right](\vec x, \vec z, \vec y) = \mathsf{R}_1\left[ \mathsf{R}[F] \right](\vec x, \vec y, \vec z)$, and thus the stable rule holds in $\mathsf{SMOOTH}$. We note that the same proof works for polynomials over any commutative semiring. 
\end{example}

With the stable rule, we can show that $\mathsf{D}$-linear transpose of the higher-order forward derivative is the higher-order reverse derivative. For convenience, by convention we define $\rho^{(0)}[f] = f = \partial^{(0)}[f]$.  

\begin{lemma}\label{lemma:Ddag} In a Cartesian $k$-reverse differential category whose reverse differential combinator satisfies the stable rule, we have that $\partial^{(n+1)}[f]^{\dagger[A \times \_ \times A^n]} = \rho^{(n+1)}[f]$, which in the term calculus is expressed as: 
\begin{align}\label{hightranspose}
   \dfrac{\mathsf{r}  \dfrac{\mathsf{d}^{(n+1)} f(x)}{\mathsf{d}x}(a_0) \cdot y \cdot a_2 \cdot \hdots \cdot a_{n+1} }{\mathsf{r}y}(0) \cdot b = \dfrac{\mathsf{r}^{(n+1)} f(x)}{\mathsf{r}x}(a_0) \cdot b \cdot a_2 \cdot \hdots \cdot a_{n+1}  
\end{align}
\end{lemma}
\begin{proof} One can show that the stable rule also holds in context, that is: 
\begin{align}
    \dfrac{\mathsf{r} \dfrac{\mathsf{d} g(c_1,x,c_2)}{x}(y) \cdot a_2 }{\mathsf{r} y} (a_1) \cdot b = \dfrac{\mathsf{r} \dfrac{\mathsf{r} g(c_1,x,c_2)}{x}(y) \cdot b }{\mathsf{d} y} (a_1) \cdot a_2
\end{align}
Now note that the stable rule can be rewritten as: 
\begin{align}
    \dfrac{\mathsf{r} \dfrac{\mathsf{d} f(x)}{x}(y) \cdot a_2 }{\mathsf{r} y} (a_1) \cdot b = \dfrac{\mathsf{r}^{(2)} f(x)}{\mathsf{d} (x)} (a_1) \cdot b \cdot \cdot a_2
\end{align}
Since the stable rule holds in context, it is straightforward to show by induction that:
\begin{align}\label{eq:induc}
    \dfrac{\mathsf{r} \dfrac{\mathsf{d}^{(n)} f(x)}{x}(y) \cdot a_1 \cdot \hdots \cdot a_{n}}{\mathsf{r} y} (a_0) \cdot b = \dfrac{\mathsf{r}^{(n+1)} f(x)}{\mathsf{r}x}(a_0) \cdot b \cdot a_2 \cdot \hdots \cdot a_{n+1}  
\end{align}
Then applying (\ref{eq:11}) to the left-hand side gives us precisely that: 
\begin{gather*}
   \dfrac{\mathsf{r}  \dfrac{\mathsf{d}^{(n+1)} f(x)}{\mathsf{d}x}(a_0) \cdot y \cdot a_2 \cdot \hdots \cdot a_{n+1} }{\mathsf{r}y}(0) \cdot b ~~\substack{= \\ \text{Def.}}~~  
   \dfrac{\mathsf{r}    \dfrac{\mathsf{d} \dfrac{\mathsf{d}^{(n)} f(x)}{x}(y) \cdot a_1 \cdot \hdots \cdot a_{n}}{\mathsf{d} y} (a_0) \cdot y }{\mathsf{r}y}(0) \cdot b \\
   \substack{= \\ \text{(\ref{eq:11})}}~~   \dfrac{\mathsf{r} \dfrac{\mathsf{d}^{(n)} f(x)}{x}(y) \cdot a_1 \cdot \hdots \cdot a_{n}}{\mathsf{r} y} (a_0) \cdot b ~~\substack{= \\ \text{(\ref{eq:induc})}}~~ \dfrac{\mathsf{r}^{(n+1)} f(x)}{\mathsf{r}x}(a_0) \cdot b \cdot a_2 \cdot \hdots \cdot a_{n+1}  
\end{gather*}
So we get that $\partial^{(n+1)}[f]^{\dagger[A \times \_ \times A^n]} = \rho^{(n+1)}[f]$ as desired. 
\end{proof}

We may finally look towards understanding the reverse Faà di Bruno's Formula. Let us first review Faà di Bruno's Formula for forward differentiation in a Cartesian differential category. To do so, we first introduce some notation. For every $n \in \mathbb{N}$, define the well-ordered set $[n+1]=\lbrace 1< \hdots < n+1 \rbrace$. Now for every subset $I = \lbrace i_1 < \hdots<  i_m \rbrace \subseteq [n+1]$, for a vector $\vec x = (x_1, \hdots, x_{n+1})$, define $x_I = (x_{i_1}, \hdots, x_{i_m})$. Lastly, we denote a \emph{non-empty} partition of $[n+1]$ as $[n+1] = S_1 \vert \hdots \vert S_k$, and let $\vert S_j \vert$ be the cardinality of $S_j$. Then Faà di Bruno's Formula \cite[Lemma 3.4]{garner2020cartesian} for the $n+1$-th derivative is given as a sum over the non-empty partitions of $[n+1]$ as follows:  
\begin{equation}\begin{gathered}\label{faa}
\dfrac{\mathsf{d}^{(n+1)} g(f(x))}{\mathsf{d}x}(a_0) \cdot a_1 \cdot a_2 \cdot \hdots \cdot a_{n+1} \\
= \sum \limits_{ [n+1]=S_1 \vert \hdots \vert S_k}  \dfrac{\mathsf{d}^{(k)} g(z)}{\mathsf{d}z}(f(a_0)) \cdot \left( \dfrac{\mathsf{d}^{\left(  \vert S_1 \vert  \right)} f(x) }{\mathsf{d}x} (a_0) \cdot a_{S_1} \right) \cdot \hdots \cdot  \left( \dfrac{\mathsf{d}^{\left(  \vert S_k \vert  \right)} f(x) }{\mathsf{d}x} (a_0) \cdot a_{S_k} \right)
\end{gathered}
\end{equation}

We are finally in a position to work out the reverse Faà di Bruno's Formula. For convenience, we will assume, without loss of generality, that for each non-empty partition $[n+1] = S_1 \vert \hdots \vert S_k$, that $1 \in [n+1]$ is always in $1 \in S_1$. As a shorthand, we write $\widehat{S_1} = S_1 - \lbrace 1 \rbrace$. 

\begin{theorem}\label{Prop:faa} In a Cartesian $k$-reverse differential category whose reverse differential combinator satisfies the stable rule, the following equality holds: 
\begin{equation}\begin{gathered}
        \dfrac{\mathsf{r}^{(n+1)} g(f(x))}{\mathsf{r}x}(a_0) \cdot b \cdot a_2 \cdot \hdots \cdot a_{n+1} = \\ 
 \sum \limits_{ \substack{[n+1]=S_1 \vert \hdots \vert S_k \\ 1 \in A_1}} \!\!\!\! \!\!\!\!\!\!\!\!\dfrac{\mathsf{r}^{\left(  \vert S_1 \vert  \right)} f(x) }{\mathsf{r}x}(a_0) \cdot \left(  \dfrac{\mathsf{r}^{(k)} g(y)}{\mathsf{r}y}(a_0) \cdot b \cdot\!\! \left( \dfrac{\mathsf{d}^{\left(  \vert S_2 \vert  \right)} f(x) }{\mathsf{d}x} (a_0) \cdot a_{S_2} \right) \cdot \hdots \cdot \left( \dfrac{\mathsf{d}^{\left(  \vert S_k \vert  \right)} f(x) }{\mathsf{d}x} (a_0) \cdot a_{S_k} \right) \!\!\right) \cdot a_{\widehat{S_1}}   
\end{gathered}
\end{equation}
\end{theorem}
\begin{proof} We compute: 
 \begin{gather*}
   \dfrac{\mathsf{r}^{(n+1)} g(f(x))}{\mathsf{r}x}(a_0) \cdot b \cdot a_2 \cdot \hdots \cdot a_{n+1} ~~\substack{= \\ \text{(\ref{hightranspose})}}~~ \dfrac{\mathsf{r}  \dfrac{\mathsf{d}^{(n+1)} g(f(x))}{\mathsf{d}x}(a_0) \cdot y \cdot a_2 \cdot \hdots \cdot a_{n+1} }{\mathsf{r}y}(0) \cdot b \\
\substack{= \\ \text{(\ref{faa})}}~~  \dfrac{\mathsf{r}  \left(\sum \limits_{ [n+1]=S_1 \vert \hdots \vert S_k}  \dfrac{\mathsf{d}^{(k)} g(z)}{\mathsf{d}z}(f(a_0)) \cdot \left( \dfrac{\mathsf{d}^{\left(  \vert S_1 \vert  \right)} f(x) }{\mathsf{d}x} (a_0) \cdot  y \cdot a_{S_1 - \lbrace 1 \rbrace} \right) \cdot \hdots \cdot  \left( \dfrac{\mathsf{d}^{\left(  \vert S_k \vert  \right)} f(x) }{\mathsf{d}x} (a_0) \cdot a_{S_k} \right) \right)  }{\mathsf{r}y}(0) \cdot b \\
\substack{= \\ \text{\textbf{[R.2]}}}~~  \sum \limits_{ [n+1]=S_1 \vert \hdots \vert S_k}  \!\!\!\! \!\!\!\!\!\!\!\! \dfrac{\mathsf{r}  \dfrac{\mathsf{d}^{(k)} g(z)}{\mathsf{d}z}(f(a_0)) \cdot \left( \dfrac{\mathsf{d}^{\left(  \vert S_1 \vert  \right)} f(x) }{\mathsf{d}x} (a_0) \cdot y \cdot a_{S_1 - \lbrace 1 \rbrace} \right) \cdot \hdots \cdot  \left( \dfrac{\mathsf{d}^{\left(  \vert S_k \vert  \right)} f(x) }{\mathsf{d}x} (a_0) \cdot a_{S_k} \right)  }{\mathsf{r}y}(0) \cdot b \\
 \substack{= \\ \text{(\ref{eq:contra})}} \sum \limits_{ \substack{[n+1]=S_1 \vert \hdots \vert S_k \\ 1 \in A_1}} \!\!\!\! \!\!\!\!\!\!\!\!  \dfrac{\mathsf{r} \dfrac{\mathsf{d}^{\left(  \vert S_1 \vert  \right)} f(x) }{\mathsf{d}x} (a_0) \cdot y \cdot a_{S_1 - \lbrace 1 \rbrace}}{\mathsf{r} y} (0) \cdot \left(   \dfrac{\mathsf{r}  \dfrac{\mathsf{d}^{(k)} g(z)}{\mathsf{d}z}(f(a_0)) \cdot z \cdot \hdots \cdot  \left( \dfrac{\mathsf{d}^{\left(  \vert S_k \vert  \right)} f(x) }{\mathsf{d}x} (a_0) \cdot a_{S_k} \right)  }{\mathsf{r}z}\left( 0 \right) \cdot b \right) \\
 \substack{= \\ \text{(\ref{hightranspose})}}  \sum \limits_{ \substack{[n+1]=S_1 \vert \hdots \vert S_k \\ 1 \in A_1}} \!\!\!\! \!\!\!\!\!\!\!\!  \dfrac{\mathsf{r}^{\left(  \vert S_1 \vert  \right)} f(x) }{\mathsf{r}x}(a_0) \cdot  \left(  \dfrac{\mathsf{r}^{(k)} g(y)}{\mathsf{r}y}(a_0) \cdot b \cdot \left( \dfrac{\mathsf{d}^{\left(  \vert S_2 \vert  \right)} f(x) }{\mathsf{d}x} (a_0) \cdot a_{S_2} \right) \cdot \hdots \cdot \left( \dfrac{\mathsf{d}^{\left(  \vert S_k \vert  \right)} f(x) }{\mathsf{d}x} (a_0) \cdot a_{S_k} \right)\!\! \right) \cdot a_{\widehat{S_1}}
\end{gather*}
So, the reverse of Faà di Bruno's Formula holds as desired. 
\end{proof}

When $n=0$ in the reverse Faà di Bruno's Formula, we get back precisely the reverse chain rule. Indeed, the only non-empty partition of $[1]$ is $[1] = S_1= \lbrace 1 \rbrace$, which is why no forward derivatives appear in the reverse chain rule. When $n=1$, the non-empty partitions of $[2]$ are $[1] = S_1= \lbrace 1 \rbrace\vert S_2=\lbrace 2 \rbrace$ and $[2] = S_1 =\lbrace 1, 2 \rbrace$. So the reverse Faà di Bruno's Formula for the second reverse derivative is:
\begin{gather*}
   \dfrac{\mathsf{r}^{(2)}g\left(f(x) \right)}{\mathsf{r}x}(a_0) \cdot b \cdot a_2 \\
   = \dfrac{\mathsf{r}^{(1)}f(x)}{\mathsf{r}x}(a_0) \cdot \left( \dfrac{\mathsf{r}^{(2)}g(y)}{\mathsf{r}y}(f(a_0)) \cdot b \cdot \left(\dfrac{\mathsf{d} f(x) }{\mathsf{d}x} (a_0) \cdot  a_2 \right) \right) + \dfrac{\mathsf{r}^{(2)}f(x)}{\mathsf{r}x}(a_0) \cdot \left( \dfrac{\mathsf{r}^{(1)}g(y)}{\mathsf{r}y}(f(a_0)) \cdot b  \right) \cdot a_2
\end{gather*}
The reverse Faà di Bruno's Formula for the third reverse derivative will involve five summands,  the formula for the fourth reverse derivative will involve fifteen summands, etc. 

\newpage 
\bibliographystyle{eptcs}
\bibliography{ACT}

\end{document}